\documentclass{article}
\usepackage{nohyperref}
\usepackage{amsmath}
\usepackage{amssymb}
\usepackage{graphicx}
\usepackage{natbib}
\usepackage{url}
\numberwithin{equation}{section}
\usepackage{ctable}
\usepackage{multicol}
\usepackage[left=0.5in,right=0.5in,bottom=0.75in,top=0.75in]{geometry}

\providecommand{\bR}{\mathbf{R}}

\providecommand{\x}{\mathbf{x}}
\providecommand{\X}{\mathbf{X}}
\providecommand{\y}{\mathbf{y}}

\renewcommand{\Pr}{\mathbb{P}}

\providecommand{\Ex}{\mathbb{E}}

\providecommand{\abs}[1]{\left\lvert#1\right\rvert}

\providecommand{\al}[2]{\begin{align}\label{#1}#2\end{align}}
\providecommand{\as}[1]{\begin{align*}#1\end{align*}}

\renewcommand{\abstract}[1]{
 \centerline{
 \begin{minipage}{0.7\linewidth}
 \hrule
 \vskip 0.1in
  \begin{center}
    {\bf Abstract}
  \end{center}
  #1
 \vskip 0.1in
 \hrule
 \end{minipage}}
 \vskip 0.3in} 

\usepackage{titling}
\pretitle{\begin{center} \LARGE \bf}
\posttitle{\par\end{center}\vskip 0.5em}


\usepackage{amsthm}
\newtheorem{thm}{Theorem}
\newtheorem{crly}{Corollary}

\title{Kernel-based aggregation of marker-level genetic association tests involving copy-number variation}

\author{Yinglei Li\\Department of Statistics\\University of Kentucky \and Patrick Breheny\\Department of Biostatistics\\University of Iowa}
\date{\today}

\begin{document}

\maketitle

\abstract{Genetic association tests involving copy-number variants (CNVs) are complicated by the fact that CNVs span multiple markers at which measurements are taken.  The power of an association test at a single marker is typically low, and it is desirable to pool information across the markers spanned by the CNV.  However, CNV boundaries are not known in advance, and the best way to proceed with this pooling is unclear.  In this article, we propose a kernel-based method for aggregation of marker-level tests and explore several aspects of its implementation.  In addition, we explore some of the theoretical aspects of marker-level test aggregation, proposing a permutation-based approach that preserves the family-wise error rate of the testing procedure, while demonstrating that several simpler alternatives fail to do so.  The empirical power of the approach is studied in a number of simulations constructed from real data involving a pharmacogenomic study of gemcitabine, and compares favorably with several competing approaches.}

\section{Introduction}

The classical genetic model states that humans have two copies of each gene, one on each chromosome.  The sequencing of individual human genomes, however, has revealed that there is an unexpected amount of structural variation present in our genetic makeup.  Sections of DNA can be deleted or duplicated, leaving individuals with fewer or more copies of portions of their genome (such individuals are said to have a {\em CNV}, or copy-number variation, at that position).  Understanding the contribution of CNVs to human variation is one of the most compelling current challenges in genetics \citep{McCarroll2008}.

As the coverage of single nucleotide polymorphism (SNP) arrays has increased, it is increasingly possibly to use this data to infer the CNV status of individuals.  Indeed, recent generations of such chips include probes specifically designed to enable measurement of copy number \citep{McCarroll2008a,Perkel2008}.  Although other technologies for determining copy number exist, the use of genotyping chips has a clear advantage in that hundreds of genome-wide association studies have been performed \citep{GWAScatalog}, and these studies may be mined for CNV data at no added cost.  Many of these studies have been carried out with very large sample sizes, thereby enabling CNV studies on a scale that would otherwise be prohibitively expensive \citep{WTCCC2007,Barnes2008}.  These factors have led to a number of studies using data from high-density genotyping arrays to investigate the nature of copy-number variation and its role in human variation and disease \citep{Redon2006,Komura2006,Simon-Sanchez2008,Cooper2011,Konishi2011}.

Two general strategies have been proposed for conducting genetic association studies of copy-number variation.  The majority of analytic techniques attempt to (1) identify or ``call'' CNVs for each individual, then (2) carry out association tests of whether individuals with a CNV differ from individuals without a CNV with regard to disease or some other phenotype.  An alternative approach is to reverse the order of those two steps: (1) carry out association testing at the single marker level, then (2) aggregate information from neighboring markers to determine CNVs associated with disease/phenotype.  The key idea of both approaches is that, because the data is noisy, it is virtually impossible to identify CNV associations from a single marker.  Because copy number variants extend over multiple markers in a sufficiently high-density array, however, we are able to carry out inferences regarding CNVs by pooling information across neighboring markers.

We refer to these two approaches, respectively, as {\em variant-level testing} and {\em marker-level testing}.  In \citet{Breheny2012}, the authors explored the relative strengths and weaknesses of the two approaches and reached the conclusion that marker-level approaches were better able to identify associations involving small, common CNVs, while variant-level approaches were better able to identify associations involving large, rare CNVs.

One serious complication with variant-level testing is that the estimated CNV boundaries from different individuals do not, in general, coincide.  This presents a number of difficulties.  Whether or not two individuals with partially overlapping CNVs should be in the same risk group for the purposes of association testing is ambiguous, and complicates both the association test itself as well as attempts to correct for multiple comparisons.  With a large sample size, the complexity of partially overlapping CNV patterns quickly becomes daunting.  Marker-level testing is an attractive alternative, as the aggregation is carried out on the test results, thereby avoiding the complications of overlapping boundaries.

\begin{figure}[htb!]
 \centering
 \includegraphics[width=0.6\linewidth]{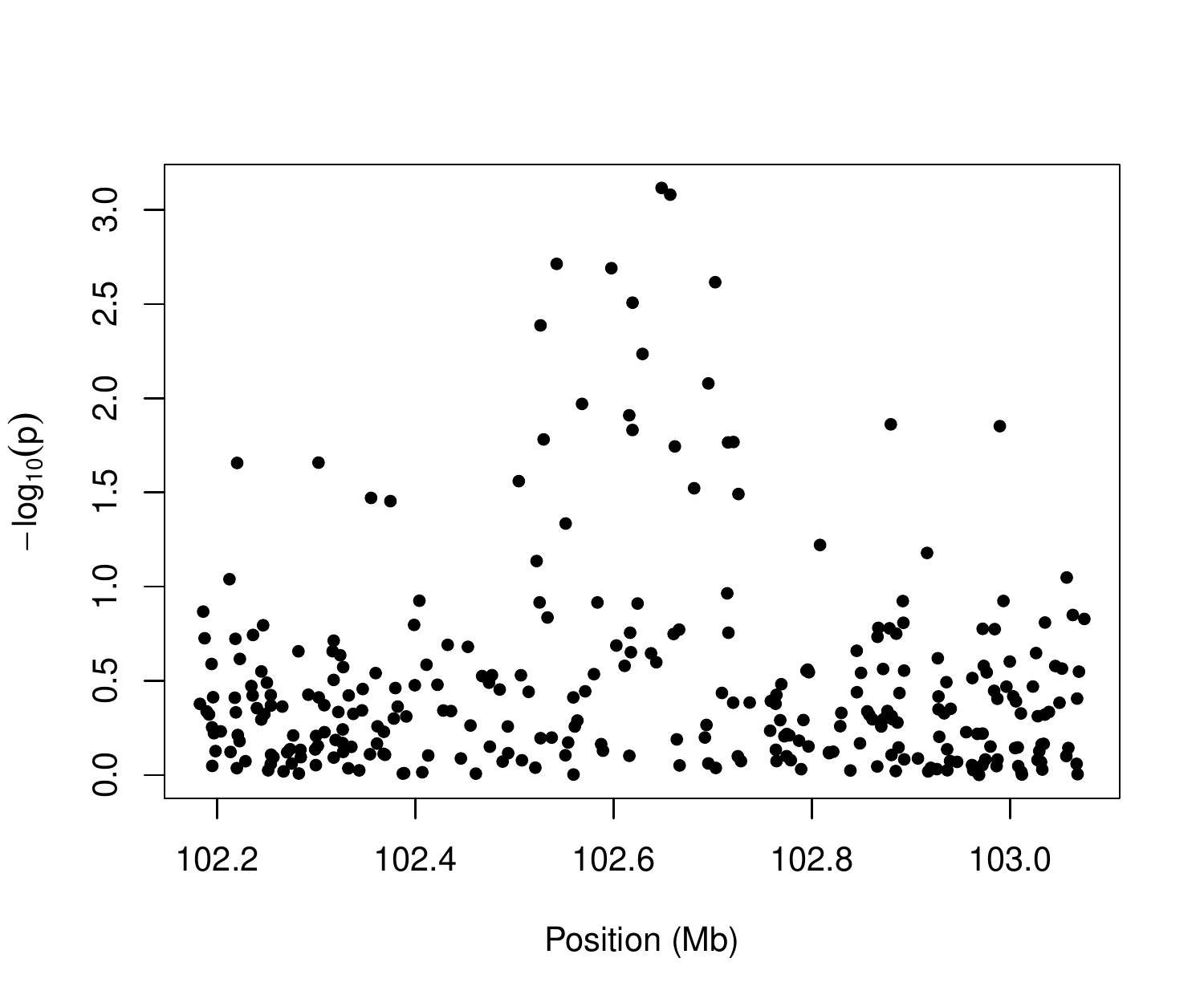}
 \caption{\label{Fig:mlt} Illustration of marker-level testing.  The $-\log_{10}(p)$ values from the marker-level tests are plotted as a function of position along the chromosome.}
\end{figure}

We illustrate the idea behind marker-level testing and aggregation in Figure~\ref{Fig:mlt}, which plots a negative log transformation of the $p$-values of the marker-level tests versus position along the chromosome.  The details of the hypothesis tests for this example are described in Section~\ref{Sec:real-data}, but the $p$-values may arise from any test for association between copy number intensity \citep{Peiffer2006} and phenotype.
%The details of the hypothesis tests for this example are given in Section~\ref{Sec:real-data}, but the $p$-values may arise from any test for association between copy number intensity \citep{Peiffer2006} and phenotype.
The salient feature of the plot is the cluster of small $p$-values between 102.5 and 102.7 Mb.  The presence of so many low $p$-values in close proximity to one another suggests an association between the phenotype and copy-number variation in that region.

To locate these clusters on a genome-wide scale, \citet{Breheny2012} used a marker-level approach based on circular binary segmentation \citep{Olshen2004,Venkatraman2007}.  Here, we take a closer look at the problem of aggregating $p$-values from marker level tests.  We present two main findings.  First, we develop a computationally efficient kernel-based approach for $p$-value aggregation.  Second, we analyze the multiple comparison properties of this approach, and of $p$-value aggregation in general.  In particular, we demonstrate that na\"ive aggregation approaches assuming exchangeability of test statistics do not preserve the family-wise error rate (FWER).  To solve this problem, we present a permutation-based approach and show that it preserves family-wise error rates while maintaining competitive power.

\section{Kernel-based aggregation}
\label{Sec:method}

Throughout, we will use $i$ to index subjects and $j$ to index markers.  Let $X_{ij}$ denote the intensity measurement for subject $i$ at marker $j$, let $y_i$ denote the phenotype for subject $i$, and let $p_j$ denote the $p$-value arising from a test of association between intensity and phenotype at marker $j$.  Finally, let $\ell_j$ denote the location of marker $j$ along the chromosome and $J$ denote the total number of markers.

Consider the aggregation
\al{eq:kernel}{T(\ell_0) = \frac{\sum_j t_j K_h(\ell_j,\ell_0)}{\sum_j K_h(\ell_j,\ell_0)},}
where $t_j = f(p_j)$ is a function of the test results for marker $j$ and $K_h(\ell_j,\ell_0)$ is a kernel that assigns a weight to $p_j$ depending on how far away marker $j$ is along the chromosome from the target location $\ell_0$.  The parameter $h$ defines the bandwidth of the kernel and thereby controls the bias-variance tradeoff --- a larger bandwidth pools test results over a larger region and thereby decreases variance, but potentially introduces bias by mixing test results beyond the boundary of a CNV with those inside the boundary.

Although in principle one could apply \eqref{eq:kernel} at any arbitrary location $\ell_0$, we restrict attention here to locations at which a marker is present and for which the bandwidth does not extent beyond the borders of the chromosome, thereby obtaining a finite set of aggregates $\{T_j\}$.  We will consider transformations $f(p_j)$ such that low $p$-values lead to large values of $t_j$, leading to significance testing based on the statistic $T=\max_j\{T_j\}$.

In this section, we describe the choice of kernel $K_h$ and transformation $f(p_j)$, as well as the issue of incorporating the direction of association for signed tests.

\subsection{Choice of kernel}
\label{Sec:kernel}

We consider two primary choices with regard to the kernel: shape and definition of bandwidth.  First, we may consider varying the shape of the kernel.  Two common choices are the flat (``boxcar'') kernel and the Epanechnikov kernel:
\al{eq:kern-flat}{
  \textrm{Flat}: \quad & K_h(\ell_j,\ell_0) = \begin{cases}
    1 &\qquad \textrm{if } \abs{\ell_j-\ell_0} \leq h\\
    0 &\qquad \textrm{otherwise}\end{cases}\\
  \textrm{Epanechnikov}: \quad & K_h(\ell_j,\ell_0) = \begin{cases} \frac{3}{4}\left\{1-\left(\frac{\ell_j-\ell_0}{h}\right)^2\right\} &\qquad \textrm{if } \abs{\ell_j-\ell_0} \leq h\\
    0 &\qquad \text{otherwise.} \end{cases}\label{eq:kern-ep}}
Intuitively, the Epanechnikov kernel would seem more attractive, as it gives higher weight to markers near the target location, and diminished weight to distant markers where bias is a larger concern.

Besides varying the shape of the kernel, we consider two definitions of bandwidth, which we refer to as {\em constant width} and {\em constant marker} (these concepts are sometimes referred to as ``metric'' and ``adaptive'' bandwidths, respectively, in the kernel smoothing literature).  In the constant width approach, as illustrated in \eqref{eq:kern-flat}-\eqref{eq:kern-ep}, the width $h$ of the kernel is constant.

In contrast, the constant marker approach expands and contracts the range of the kernel as needed so that there are always $k$ markers in the positive support of the kernel.  Specifically, $h_k(\ell_0) = \abs{\ell_0-\ell_{[k]}}$, where $\ell_{[k]}$ is the location of the $k$th closest marker to $x$.  For the constant width approach, the number of markers given positive weight by the kernel varies depending on $\ell_0$.

The general tradeoff between the two approaches is that as we vary the target location, $\ell_0$, constant width kernels suffer from fluctuating variance because the effective sample size is not constant, whereas constant marker kernels suffer from fluctuating bias because the size of the region over which test results are pooled is not constant.  We investigate the benefits and drawbacks of these various kernels in Section~\ref{Sec:sim}.

%We consider primarily a nearest-neighbor version of the flat or ``boxcar'' kernel:
%\al{eq:nn}{K_h(\ell_j,\ell_0) = \begin{cases}
%1 &\qquad \textrm{if } r_j \leq h\\
%0 &\qquad \textrm{otherwise}\end{cases},}
%where $r_j$ is the rank of $\abs{\ell_j-\ell_0}$.  In other words, the kernel assigns a weight of 1 if the marker is one of the $h$ nearest neighbors to $\ell_0$ and 0 otherwise.

As a point of reference, the flat, constant marker kernel is similar to the simple moving average, although not exactly the same.  For example, consider the following illustration.
\begin{center}
\includegraphics[scale=0.75]{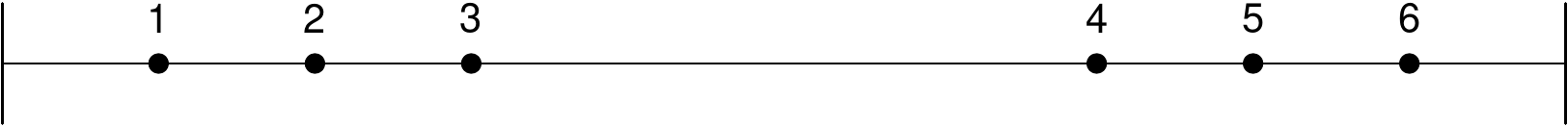}
\end{center}
Suppose $h=3$.  At $\ell_3$, the three nearest neighbors are $\{p_1,p_2,p_3\}$, while at $\ell_4$, the three nearest neighbors are $\{p_4,p_5,p_6\}$.  Thus, combinations such as $\{p_3,p_4,p_5\}$ are not considered by the kernel approach.  This prevents the method from aggregating test results over inappropriately disperse regions of the chromosome, such as across the centromere.

\subsection{Transformations}
\label{Sec:trans}

As suggested by \eqref{eq:kernel}, directly pooling $p$-values is not necessarily optimal.  Various transformations of $p$ may be able to better discriminate true associations from noise.  Specifically, we consider the following transformations:
\al{eq:trans-up}{\textrm{p}: \quad t_j &= 1-p_j\\
\label{eq:trans-uz}\textrm{Z}: \quad t_j &= \Phi^{-1}(1-p_j)\\
\label{eq:trans-ul}\textrm{log}: \quad t_j &= -\log p_j,}
where the text to the left of the equation is the label with which we will refer to these transformations in later figures and tables.  The transformations are constructed in such a way that low $p$-values produce high values of $t_j$ for all three transformations.

All three transformations have a long history in the field of combining $p$-values.  Forming a combination test statistic based on the sum of $\log p$ values (or equivalently, the log of the product of the $p$-values) was first proposed by \citet{Fisher1925} --- the so-called Fisher combination test.  The transformation \eqref{eq:trans-uz} was proposed by \citet{Stouffer1949}, who also studied the properties of sums of these normal-transformed $p$-values.  Finally, \eqref{eq:trans-up} was proposed and studied by \citet{Edgington1972}.  Several authors have followed up on these proposals \citep{Littell1971,Hedges1985,Feller1968,Good1955}.  Throughout this literature, the majority of work has focused on these scales --- uniform, Gaussian, and logarithmic --- each of which has been shown to have advantages and drawbacks.  There is no uniformly most powerful method of combining $p$-values \citep{Zaykin2002}.

The present application differs from the classical work described above in that the borders of the CNVs are not known, and thus neither is the appropriate set of $p$-values to combine.  Consequently, we must calculate many combinations $\{T_j\}$, which are partially overlapping and therefore not independent, thereby requiring further methodological extensions.  The implications of these concerns is addressed in Section~\ref{Sec:fwer}.

\subsection{Direction of association}
\label{Sec:sign}

Some association tests ($z$ tests, $t$ tests) have a direction associated with them, while others ($\chi^2$ tests, $F$ tests) do not.  As we will see in Section~\ref{Sec:sim}, it is advantageous to incorporate this direction into the analysis when it is available, as it diminishes noise and improves detection.  We introduce here extensions of the transformations presented in Section~\ref{Sec:trans} that include the direction of association.

Let $s_j$ denote the direction of association for test $j$.  For example, in a case control study, if intensities were higher for cases than controls at marker $j$, then $s_j=1$.  At markers where CNV intensities were higher for controls than cases, $s_j=-1$.  The signs are arbitrary; their purpose is to reflect the fact that an underlying, latent CNV affects both phenotype and intensity measures; thus, switching directions of association are inconsistent with the biological mechanism being studied and likely to be noise.

When $s_j$ is available, we adapt the three transformations from \ref{Sec:trans} as follows:
\al{eq:trans-sp}{\textrm{p}: \quad t_j &= s_j(1-p_j)\\
\label{eq:trans-sz}\textrm{Z}: \quad t_j &= \Phi^{-1}\left(\frac{1+s_j(1-p_j)}{2}\right)\\
\label{eq:trans-sl}\textrm{log}: \quad t_j &= -s_j\log p_j.}
All of these transformations have the same effect: when $p_j \approx 0$ and $s_j=1$, $t_j \gg 0$; when $p_j \approx 0$ and $s_j=-1$, $t_j \ll 0$; and when $p_j \approx 1$, $t_j \approx 0$ regardless of the value of $s_j$.  In other words, the test results combine to give an aggregate value $T(\ell_0)$ that is large in absolute value only if the test results have low $p$-values and are consistently in the same direction.

\section{Significance testing and FWER control}
\label{Sec:fwer}

\subsection{Exchangeability}
\label{Sec:exch}

In any analysis that involves aggregating marker-level test results, it is of interest to be able to quantify the significance of regions like those depicted in Figure~\ref{Fig:mlt}.  This is not trivial, however, as the lack of exchangeability between test results complicates matters and causes various na\"ive approaches to fail.  In this section, we illustrate the consequences of non-exchangeability by comparing three approaches to establishing the combined significance of a region with a preponderance of low $p$-values.

One approach, suggested in \citet{Breheny2012}, is to use circular binary segmentation (CBS; implemented in the \texttt{R} package \texttt{DNAcopy}).  This method aggregates neighboring $p$-values by calculating the $t$-test statistic comparing the intensity of a given region with that of the surrounding region.  The significance of this test statistic is quantified by comparing it to the distribution of maximum test statistics obtained by permuting the $\{p_j\}$ values \citep{Olshen2004,Venkatraman2007}.  Crucially, however, this approach assumes that the test results $\{p_j\}$ are exchangeable as the justification for permuting them.

Alternatively, we may use the kernel methods described in Section~\ref{Sec:kernel} to aggregate the neighboring test results, thereby obtaining $T_{\max}=\max_j\{T_j\}$.  One approach to approximating the null distribution of $T_{\max}$ is to use Monte Carlo integration based on the fact that, under the null distribution, $p_j \sim \textrm{Uniform}(0,1)$.  % and $s_j$ equals $+1$ or $-1$ with equal probability.  
Thus, for any choice of transformation and kernel in \eqref{eq:kernel}, we may generate an arbitrary number of independent draws $\{T_{\max}^{(b)}\}_{b=1}^B$ from the null distribution function $F_0$ of $T_{\max}$, and use the empirical CDF of those draws to obtain the estimate $\hat{F}_0$.  Thus, we obtain a test for the presence of a CNV-phenotype association based on $p=1-\hat{F}_0(T_{\max})$.  The crucial assumption here is that, under the null, the $p$-values are independent.

An alternative to the Monte Carlo approach for quantifying the significance of $T_{\max}$, described fully in Section~\ref{Sec:perm}, involves obtaining $\hat{F}_0$ by permuting the phenotype prior to aggregation of the marker-level tests.  %In what follows, we will refer to these three approaches as CBS (circular binary segmentation), KMC (kernel Monte Carlo), and KP (kernel permutation).

\begin{table}[htb!]
\begin{center}
\caption{\label{Tab:fwer} Preservation of Type I error for three methods with nominal $\alpha=.05$ in two possible settings for which the null hypothesis holds.  The simulated genomic region contained 200 markers, 30 of which were spanned by a CNV.  The CNV was present in either 0\% or 50\% of the samples, depending on the setting.  A detailed description of the simulated data is given in Section~\ref{Sec:sim}.}
\begin{tabular}{@{}lp{2cm}p{2cm}p{2cm}@{}}
\toprule
 & Circular & Kernel & Kernel \\
 & binary   & Monte  & Permutation\\
 & segmentation & Carlo\\
\midrule
No CNV         & 0.05 & 0.06 & 0.06 \\ 
No Association & 0.20 & 0.54 & 0.06 \\ 
\bottomrule
\end{tabular}
\end{center}
\end{table}

Consider a genomic region in which individuals may have a CNV.  The goal of the analysis is to detect CNVs associated with a particular phenotype.  Thus, the null hypothesis may hold in one of two ways: (1, ``No CNV'') no individuals with CNVs are present in the sample, or (2, ``No association'') individuals with CNVs are present in the sample, but the CNV does not change the probability of developing the phenotype.  Table~\ref{Tab:fwer} demonstrates that while all three methods have the proper type I error rate under null hypothesis 1 (No CNV), only the permutation approach preserves the correct type I error in the case where a CNV is present, but not associated with the disease.  This is due to the fact that when a CNV is present, although it is still true that the marginal distribution of each $p_j$ is Uniform(0,1), the CNV introduces correlation between nearby markers, thereby violating the assumptions of exchangeability and independence made by the CBS and Monte Carlo approaches.  This phenomenon is also illustrated graphically in Figure~\ref{Fig:fwer}.

\begin{figure}[htb!]
 \centering
 \includegraphics[width=0.75\linewidth]{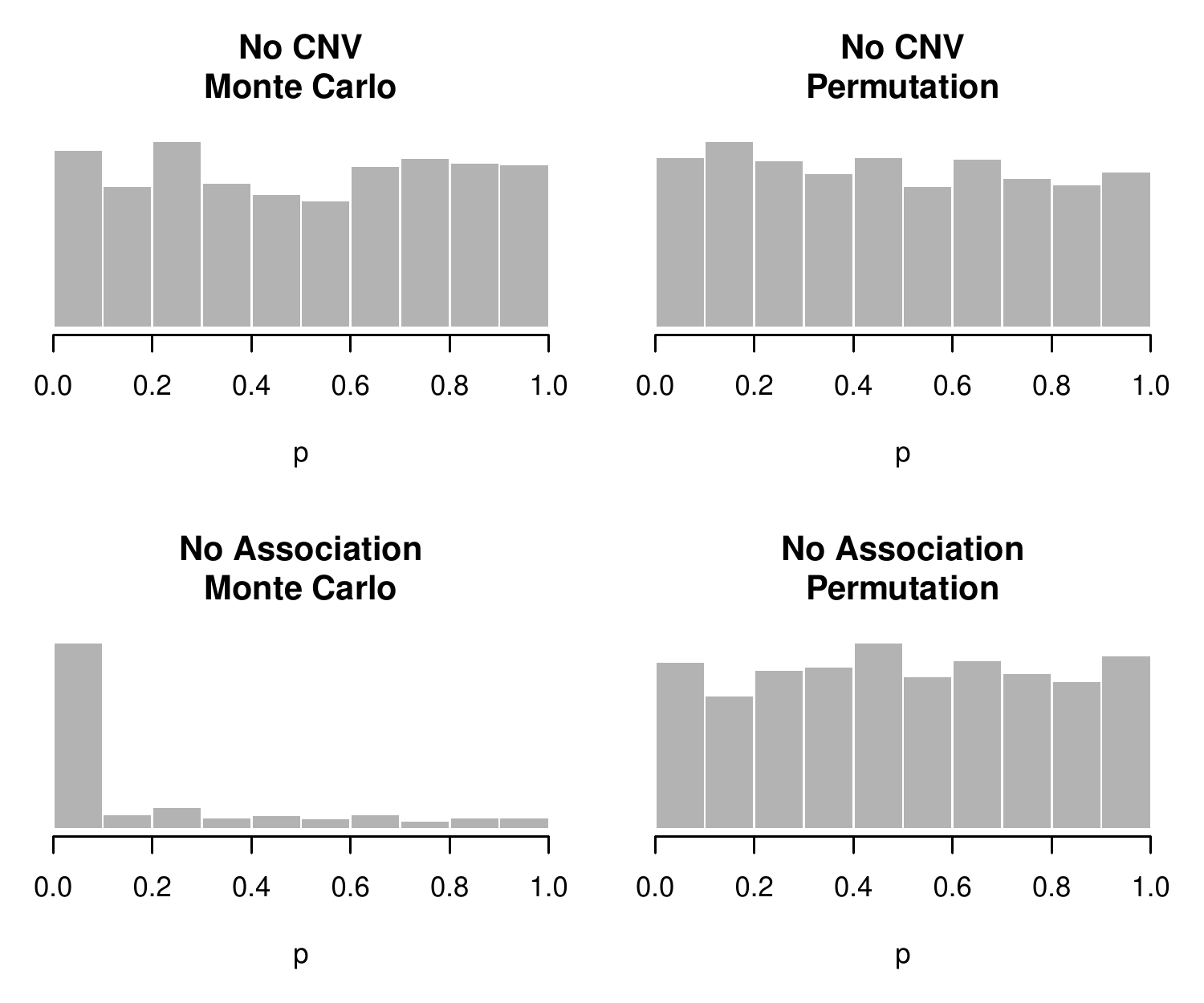}
 \caption{\label{Fig:fwer} Ability (or inability) of Monte Carlo and Permutation approaches to maintain family-wise error rate under the two null scenarios.  The implementation of CBS provided by \texttt{DNAcopy} does not return $p$-values (only whether they fall above or below a cutoff), and thus could not be included in this plot.  Valid $p$-values should appear uniformly distributed under the null; this is clearly not the case in the lower left histogram.}
\end{figure}

We make the following additional observations: (1) The CBS approach is somewhat more robust to the exchangeability issue than the Monte Carlo approach; {\em i.e.}, its type I error rate is not as badly violated.  (2) The data simulated here for the ``no association'' setting are somewhat exaggerated: the CNV was present in 50\% of the population and the signal to noise ratio was about twice as high as that typically observed in real data.  In more realistic settings, the violation of type I error rate is not nearly as severe.  The results in Table~\ref{Tab:fwer} and Figure~\ref{Fig:fwer} are intended to be an illustrative counterexample to demonstrate that CBS and kernel Monte Carlo are not guaranteed to preserve the type I error in all settings.  (3) Circular binary segmentation was developed for the purpose of detecting CNVs, not aggregating marker-level tests, and its failure to preserve the family-wise error rate in this setting is in no way a criticism of CBS in general.

\subsection{Permutation approach}
\label{Sec:perm}

We now formally define the kernel permutation method introduced in Section~\ref{Sec:exch} and show that it preserves the family-wise error rate for the problem of CNV association testing.  For a given set of test results $\{p_j\}$, consider quantifying whether or not the data represent compelling evidence for a CNV-phenotype association using the statistic
\al{eq:tstat}{T_{\max}=\max_j\{T_j\}.}
If the tests are directional, with results $\{p_j,s_j\}$, we use $T_{\max}=\max_j\{\abs{T_j}\}$.
% One-sided tests?

To obtain the null distribution of $T_{\max}$, we use a permutation approach, generating up to $n!$ unique draws $\{T_{\max}^{(b)}\}_{b=1}^B$ from the permutation distribution of $T_{\max}$.  The procedure is as follows.  At any given iteration, draw a random vector of phenotypes $\y^{(b)}$ by permuting the original vector of phenotypes.  Next, carry out marker level tests of association between the original CNV intensities and the permuted vector of phenotypes, obtaining a vector of permutation test results $\{p_j^{(b)}\}$.  Finally, apply the kernel aggregation procedure described in Section~\ref{Sec:kernel} to obtain $\{T_j^{(b)}\}$ and $T_{\max}^{(b)}$.

We may then use the empirical CDF of these draws from the permutation distribution of $T_{\max}$ to obtain the estimate $\hat{F}_0$.  Thus, we obtain a global test for the presence of a CNV-phenotype association based on $p=1-\hat{F}_0(T_{\max})$.  By preserving the correlation structure of the original CNV intensities, this approach does not rely on any assumptions of exchangeability or independence across neighboring markers, and is thereby able to preserve the type I error rate of the testing procedure, unlike the other approaches described in Section~\ref{Sec:exch}.  We now formally present this result, the proof of which appears in the Appendix.

\begin{thm}
\label{Thm:fwer}
Let $H_0$ denote the hypothesis that the phenotype, $y_i$, is independent of the vector of CNV intensities, $\x_i$.  Then, using the permutation approach described above with any of the kernel aggregation approaches in Section~\ref{Sec:kernel}, for any $\alpha \in (0,1)$,
\al{eq:thm}{\Pr(\textrm{Type I error}) \leq \alpha.}
\end{thm}

It is worth pointing out that the above theorem is proven for the case in which all permutations of $\{y_i\}$ are considered.  In practice, as it is usually impractical to consider all permutations, only a random subset of these permutations are considered.  However, by the law of large numbers, the above conclusion still holds approximately, and may be made as precise as necessary by increasing the value of $B$, the number of permutations evaluated.  For the numerical results in Section~\ref{Sec:sim},
% and \ref{Sec:real-data},
we use $B=1,000$. 

The global test above is of limited practical benefit in the sense that it does not indicate the location of the associated CNV.  Thus, we also consider the following equivalent marker-level test: declare significant evidence for the presence of a CNV-phenotype association at any marker for which $T_j > F_0^{-1}(1-\alpha)$.  Below, we state the corollary to Theorem~\ref{Thm:fwer} for the kernel permutation method, viewed as a multiple testing procedure for each marker.

\begin{crly}
\label{Crly:fwer}
Let $H_{0j}$ denote the hypothesis that the phenotype, $y_i$, is independent of the CNV intensity at marker $j$, $X_{ij}$.  Then, under the global null hypothesis that $y_i$ is jointly independent of $\{X_{ij}\}$, for any $\alpha \in (0,1)$,
\al{eq:crly}{\Pr(\textrm{At least one Type I error}) \leq \alpha}
using the permutation approach described above and $T_j > F_0^{-1}(1-\alpha)$ as the test function for $H_{0j}$.  In other words, the testing procedure described above controls the FWER in the weak sense at level $\alpha$.
\end{crly}

It is worth noting that the procedure controls the FWER only in the weak sense --- in other words, that it limits the probability of a false declaration of a CNV only under the global null hypothesis that there are no CNVs associated with the outcome.  Typically in multiple testing scenarios, strong control is desirable.  However, in the case of CNV-phenotype association, strong control is impractical, as it would imply that a method not only identifies CNV-phenotype associations, but can perfectly detect the genomic boundary of any associated CNV.  This is an unrealistic requirement; in practice, there is no way to prevent the possibility that a detected CNV-phenotype association may spill over beyond the boundary of the CNV.

\section{Gemcitabine study}
\label{Sec:real-data}

In this section we describe a pharmacogenomic study of gemcitabine, a commonly used treatment for pancreatic cancer.  We begin by describing the design of the study \citep[this description is similar to that provided in][]{Breheny2012}, then analyze data from the study using the proposed kernel-based aggregation method.  This data will also be used to simulate measurement errors for the simulation studies in Section~\ref{Sec:sim}.

The gemcitabine study was carried out on the Human Variation Panel, a model system consisting of cell lines derived from Caucasian, African-American and Han Chinese-American subjects (Coriell Institute, Camden, NJ).  Gemcitabine cytotoxicity assays were performed at eight drug dosages (1000, 100, 10, 1, 0.1, 0.01, 0.001, and 0.0001 uM) \citep{Li2008c}.  Estimation of the phenotype IC50 (the effective dose that kills 50\% of the cells) was then completed using a four parameter logistic model \citep{Davidian1995}.  Marker intensity data for the cell lines was collected using the Illumina HumanHap 550K and HumanHap510S at the Genotyping Shared Resources at the Mayo Clinic in Rochester, MN, which consists of a total of 1,055,048 markers \citep{Li2009a,Niu2010}.  Raw data was normalized according to the procedure outlined in \citet{Barnes2008}.

172 cell lines (60 Caucasian, 53 African-American, 59 Han Chinese-American) had both gemcitabine cytotoxicity measurements and genome-wide marker intensity data. To illustrate the application of the kernel-based aggregation approach, we selected one chromosome (chromosome 3) from the genome-wide data.  To control for the possibility of population stratification, which can lead to spurious associations, we used the method developed by \citet{Price2006}, which uses a principal components analysis (PCA) to adjust for stratification.  At each marker, a linear regression model was fit with PCA-adjusted IC50 as the outcome and intensity at that marker as the explanatory variable; these models produce the marker-level tests.

\begin{figure}[htb!]
 \centering
 \includegraphics[width=0.6\linewidth]{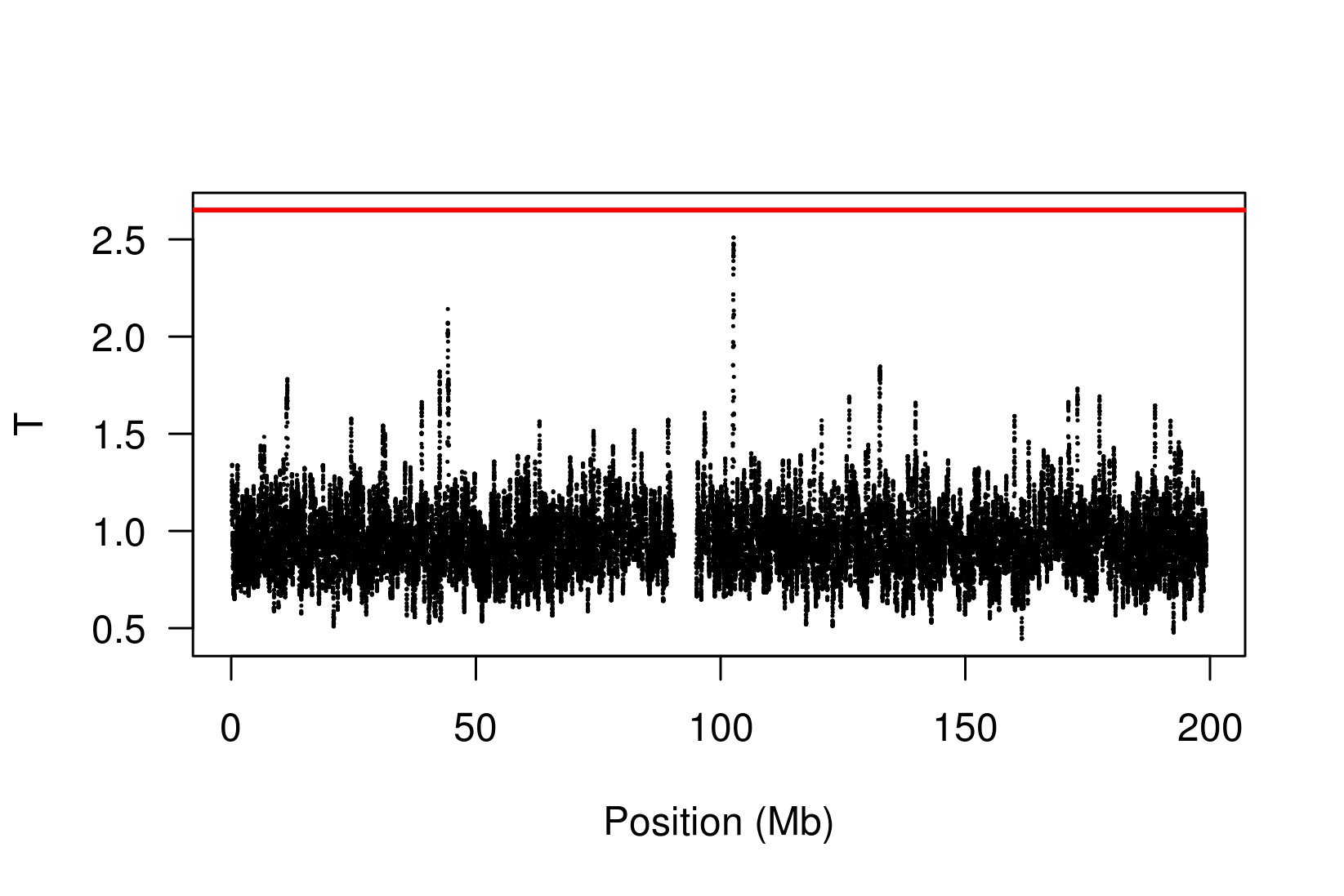}
 \caption{\label{Fig:gemc} Analysis of the gemcitabine data (Chromosome 3) using the proposed kernel aggregation method.  The kernel aggregations $T_j$ are plotted against chromosomal position.  The red line indicates the cutoff for chromosome-wide FWER significance at the $\alpha=.1$ level.}
\end{figure}

We analyzed these data using the kernel-based approach described in Section~\ref{Sec:method} with a bandwidth of 50 markers and the log transformation.  The results are shown in Figure~\ref{Fig:gemc}.  Note the presence of a peak at 102.6 Mb; this genomic region was also illustrated in Figure~\ref{Fig:mlt}.  The red line indicates the FWER-controlled, chromosome-wide significance threshold at the $\alpha=0.1$ level.  As the figure indicates, there is insufficient evidence in this study to establish a CNV association involving response to gemcitabine ($p=0.16$) after controlling the chromosome-wide FWER.

Copy number variation in the region of chromosome 3 at 102.6 Mb, which is in close proximity to the gene ZPLD1, has been found by \citet{Glessner2010} to be associated with childhood obesity.  An earlier analysis of this data by \citet{Breheny2012} indicated suggestive evidence that this region harbors a CNV association with gemcitabine response, but lacked a formal way to control the error rate at the chromosome-wide level.  This example illustrates the need for the more rigorous approach we develop here.  The lack of significance in this example is perhaps not surprising, in that 172 subjects is a relatively small sample size for a CNV association study.

\section{Simulations}
\label{Sec:sim}

\subsection{Design of spike-in simulations}

In this section, we study the ability of the proposed approach to detect CNV-phenotype associations using simulated CNVs and corresponding intensity measurements.  The validity of our conclusions depends on how realistic the simulated data is, so we have given careful thought to simulating this data in as realistic a manner as possible.  The spike-in design that we describe here is also described in \citet{Breheny2012}.

The basic design of our simulations is to use real data from the gemcitabine study described in Section~\ref{Sec:real-data}, ``spike'' a signal into it, then observe the frequency with which we can recover that signal.  We used circular binary segmentation \citep{Olshen2004,Venkatraman2007} to estimate each sample's underlying mean intensity at every position along the chromosome, then subtracted the actual intensity measurement from the estimated mean to obtain a matrix of residuals representing measurement error.  This matrix, denoted $\bR$, has 172 rows (one for each cell line) and 70,542 columns (one for each marker).

We then used these residuals to simulate noise over short genomic regions in which a single simulated CNV is either present or absent.  Letting $i$ denote subjects and $j$ denote markers, the following variables are generated: $z_i$, an indicator for the presence or absence of a CNV in individual $i$; $x_{ij}$, the intensity measurement at marker $j$ for individual $i$; and $y_i$, the phenotype.  For the sake of clarity, we focus here on a random sampling design in which the outcome is continuous; similar results were obtained from a case-control sampling design in which the outcome is binary.  In the random sampling design, the CNV indicator, $z_i$, is generated from a Bernoulli distribution, where $\gamma=\Pr(z_i=1)$ is the frequency of the CNV in the population; subsequently, $y_i|z_i$ is generated from a normal distribution whose mean depends on $z_i$.

\begin{figure}[b!]
 \centering
 \includegraphics[width=\linewidth]{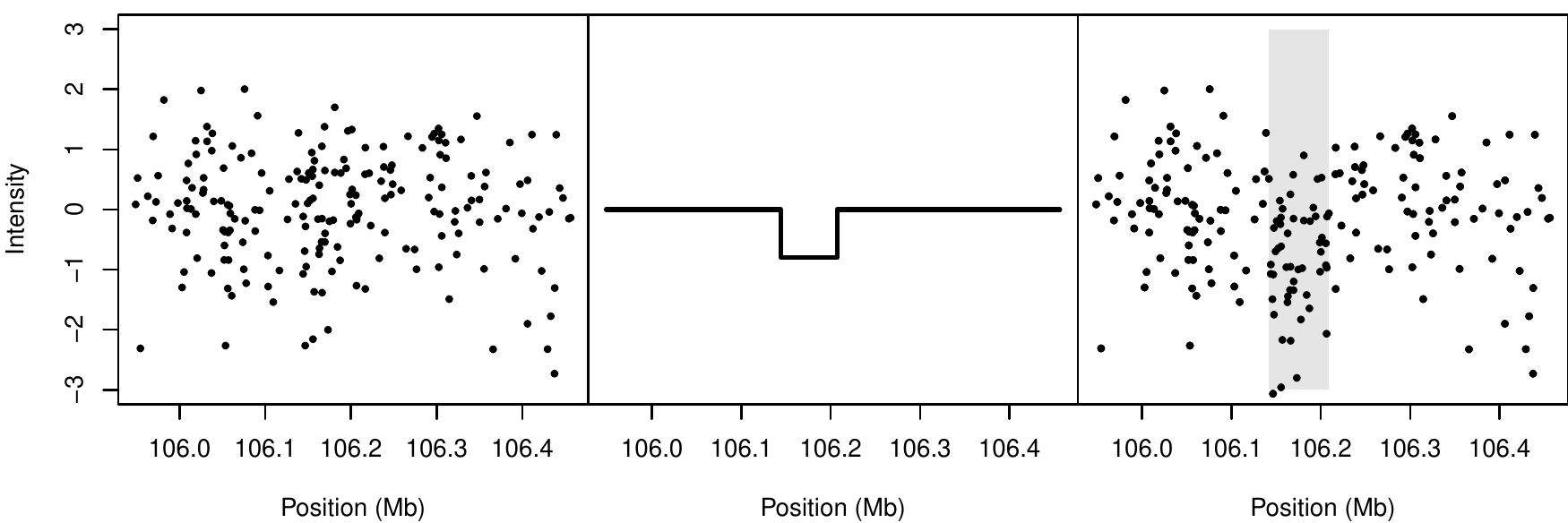}
 \caption{\label{Fig:sim} Illustration of spike-in simulation design.  {\em Left:} The noise, randomly drawn from among the estimated measurement errors for a single cell line.  {\em Middle:} The spiked-in signal.  {\em Right:} The resulting simulated data; the gray shaded region denotes the boundary of the spiked-in CNV.}
\end{figure}

At the beginning of each simulated data set, 200 markers were randomly selected from the columns of $\bR$.  The measurement error for simulated subject $i$ was then drawn from the observed measurement errors at those markers for a randomly chosen row of $\bR$.  Thus, within a simulated data set, all subjects are studied with respect to the same genetic markers, but the markers vary from data set to data set.  Simulating the data in this way preserves all the features of outliers, heavy-tailed distributions, skewness, unequal variability among markers, and unequal variability among subjects that are present in real data.

The intensity measurements $\{x_{ij}\}$ derive from these randomly sampled residuals.  To the noise, we add a signal that depends on the presence of the simulated CNV, $z_i$.  The added signal is equal to zero unless the simulated genome contains a CNV encompassing the $j$th marker; otherwise the added signal is equal to the standard deviation of the measurement error times the signal to noise ratio.  Our simulations employed a signal-to-noise ratio of 0.8, which corresponded roughly to a medium-sized detectable signal based on our inspection of the gemcitabine data.  Note that the phenotype and intensity measurement are conditionally independent given the latent copy-number status $z_i$.  An illustration of the spike-in process is given in Figure~\ref{Fig:sim}.

For the Illumina Human1M-Duo BeadChip, which has a median spacing of 1.5 kb between markers, 200 markers corresponds to simulating a 300 kb genomic region.  We varied the length of the CNV from 10 to 50 markers, corresponding to a size range of 15 to 75 kb.  For the simulations presented in the remainder of the article, we used a sample size of $n=1,000$ and an effect size (change in mean divided by standard deviation) for the continuous outcome of 0.4.

\subsection{Transformations}
\label{Sec:trans-results}

We begin by examining the impact on power of the various transformations proposed in Sections~\ref{Sec:trans} and \ref{Sec:sign}.  In order to isolate the effect of transformation, we focus here on the ``optimal bandwidth'' results: the bandwidth of the kernel was chosen to match the number of markers in the underlying CNV.  This will lead to the maximum power to detect a CNV-phenotype association, although this approach is clearly not feasible in practice, as the size of an underlying CNV is unknown.

\begin{figure}[htb!]
 \centering
 \includegraphics[width=\linewidth]{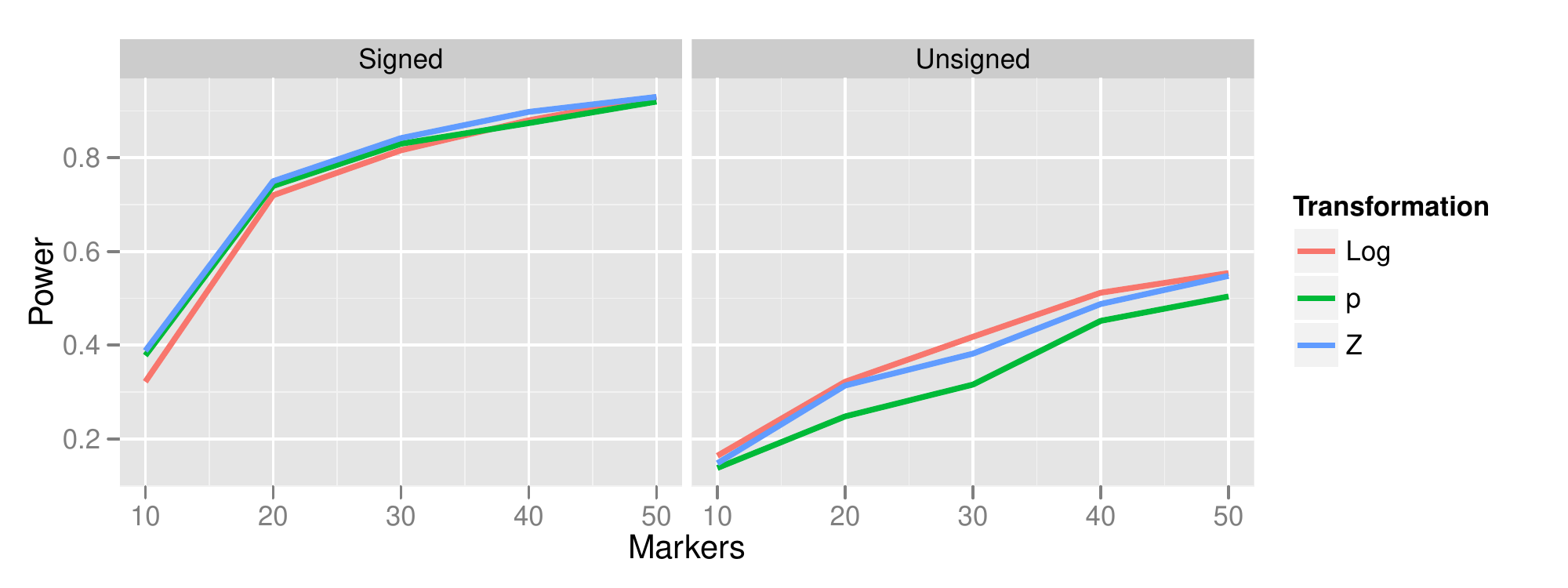}
 \caption{\label{Fig:trans} Effect of transformation choice on power.  Population CNV frequency was set to 10\%; optimal bandwidths used.}
\end{figure}

The relationship between power and transformation choice is illustrated in Figure~\ref{Fig:trans}.  The figure illustrates a basic trend that held consistently over many CNV frequencies and bandwidth choices: although the various transformations do not dramatically alter power, the normalizing transformation (Z) is most powerful for signed test results, while the log transformation is most powerful for unsigned test results.  In the results that follow, unless otherwise specified, we employ the normalizing transformation for signed test results and the log transformation for unsigned tests.  The substantial gain in power attained by incorporating the direction of association is also apparent from Figure~\ref{Fig:trans} by comparing the left and right halves of the figure.  

\subsection{Choice of kernel}

In this section, we examine two aspects of kernel choice: bandwidth implementation (constant-width vs. constant-marker) and kernel shape (flat vs. Epanechnikov), defined in Section~\ref{Sec:kernel}.  When all markers are equally spaced, the constant-width and constant-marker kernels are equivalent.  To examine the impact on power when markers are unequally spaced, we selected at random a 200-marker sequence from chromosome 3 of the Illumina HumanHap 550K genotyping chip and spiked in CNVs of various sizes.  The optimal bandwidth (either in terms of the number of markers or base pairs spanned by the underlying CNV) was chosen for each method.

The left side of Figure~\ref{Fig:kernel} presents the results of this simulation.  The constant-marker approach is substantially more powerful.  When the number of markers is not held constant, the aggregation measure $T_j$ is more highly variable for some values of $j$ than others.  This causes the null distribution of $T_{\max}$ to have thicker tails, which in turn increases the $p$-value for the observed $T_{\max}$, thus lowering power.  This phenomenon manifests itself most dramatically for small bandwidths.  Consequently, throughout the rest of this article, we employ constant-marker kernels for all analyses.

\begin{figure}[htb!]
 \centering
 \includegraphics[width=\linewidth]{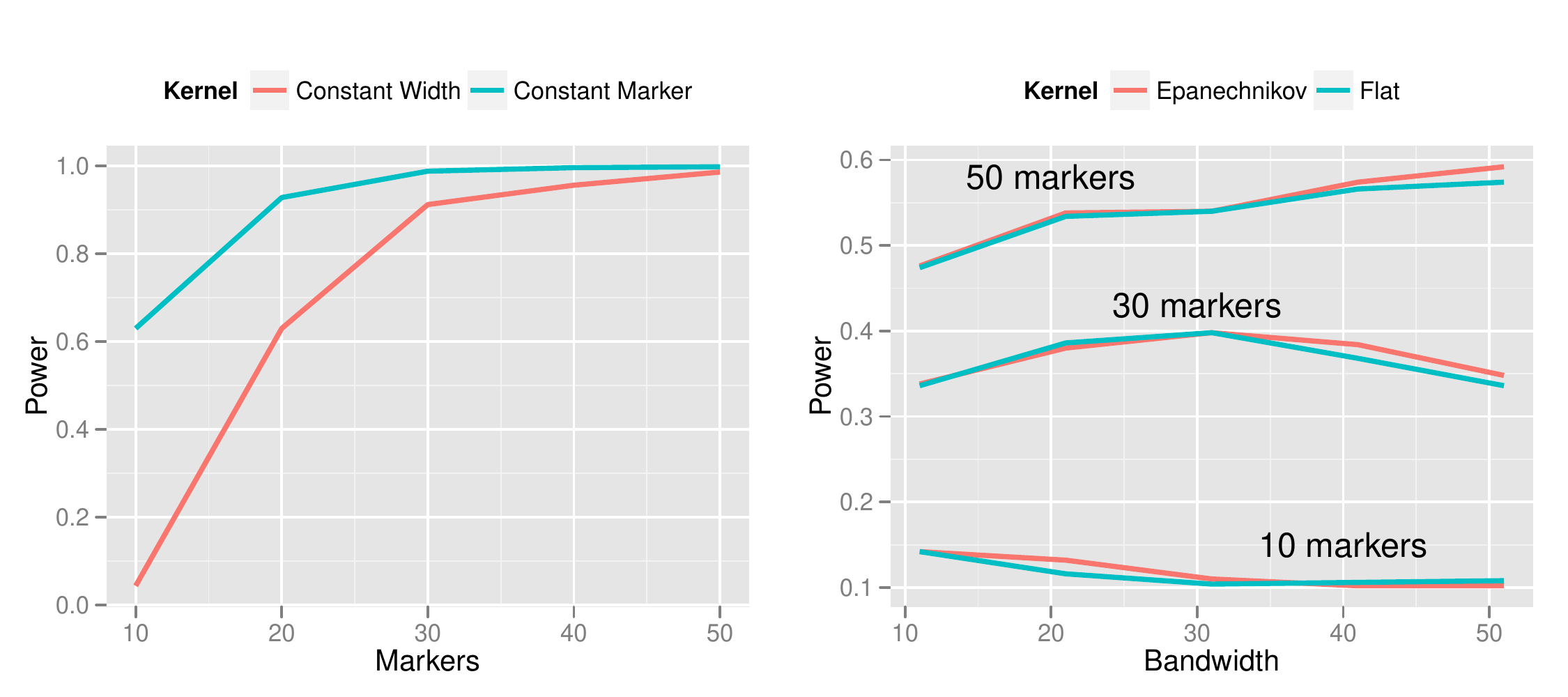}
 \caption{\label{Fig:kernel} Effect of kernel choice on power.  {\em Left:} Constant-width kernel vs. constant-marker kernel.  {\em Right:} Flat vs. Epanechnikov kernel.  In both plots, population CNV frequency was 10\%, test results were unsigned, and the log transformation was used.}
\end{figure}

The right side of Figure~\ref{Fig:kernel} presents the results of changing the kernel shape from the flat kernel described in \eqref{eq:kern-flat} to the Epanechnikov kernel described in \eqref{eq:kern-ep}.  We make several observations: (1) The shape of the kernel has little impact on power; the two lines are nearly superimposed.  (2) The kernel approach is relatively robust to choice of bandwidth; even 5-fold differences between the bandwidth and optimal bandwidth do not dramatically decrease power.  (3) Nevertheless, the optimal bandwidth does indeed occur when the number of markers included in the kernel matches the true number of markers spanned by the CNV.  (4) The Epanechnikov kernel is slightly more robust to choosing a bandwidth that is too large than the flat kernel is.  This makes sense, as the Epanechnikov kernel gives less weight to the periphery of the kernel.

%\subsection{Signed vs. unsigned aggregation}
%\label{Sec:sign-results}

%Figure~\ref{Fig:signed} illustrates the considerable impact upon power of incorporating sign information.  The substantial increase in power reinforces the common-sense intuition that when directional information about the test results is available, it is advantageous to take this information into account.  Figure~\ref{Fig:signed} presents results for the log transformation; the normalizing transformations of \eqref{eq:trans-uz} and \eqref{eq:trans-sz} --- which have higher power for the signed results and lower power for the unsigned results --- would show an even larger gain.

\subsection{Kernel-based aggregation vs. variant-level testing}

Lastly, we compare the kernel-based aggregation approach with variant-level testing.  To implement variant-level testing, each sample was assigned a group (``variant present'' or ``variant absent'') on the basis of whether a CNV was detected by CBS.  A two-sample $t$-test was then carried out to test for association of the CNV with the phenotype.  This variant-level approach was compared with kernel-based aggregation of marker-level testing for a variety of bandwidths.  The results are presented in Figure~\ref{Fig:vlt-mlt}.

\begin{figure}[htb!]
 \centering
 \includegraphics[width=\linewidth]{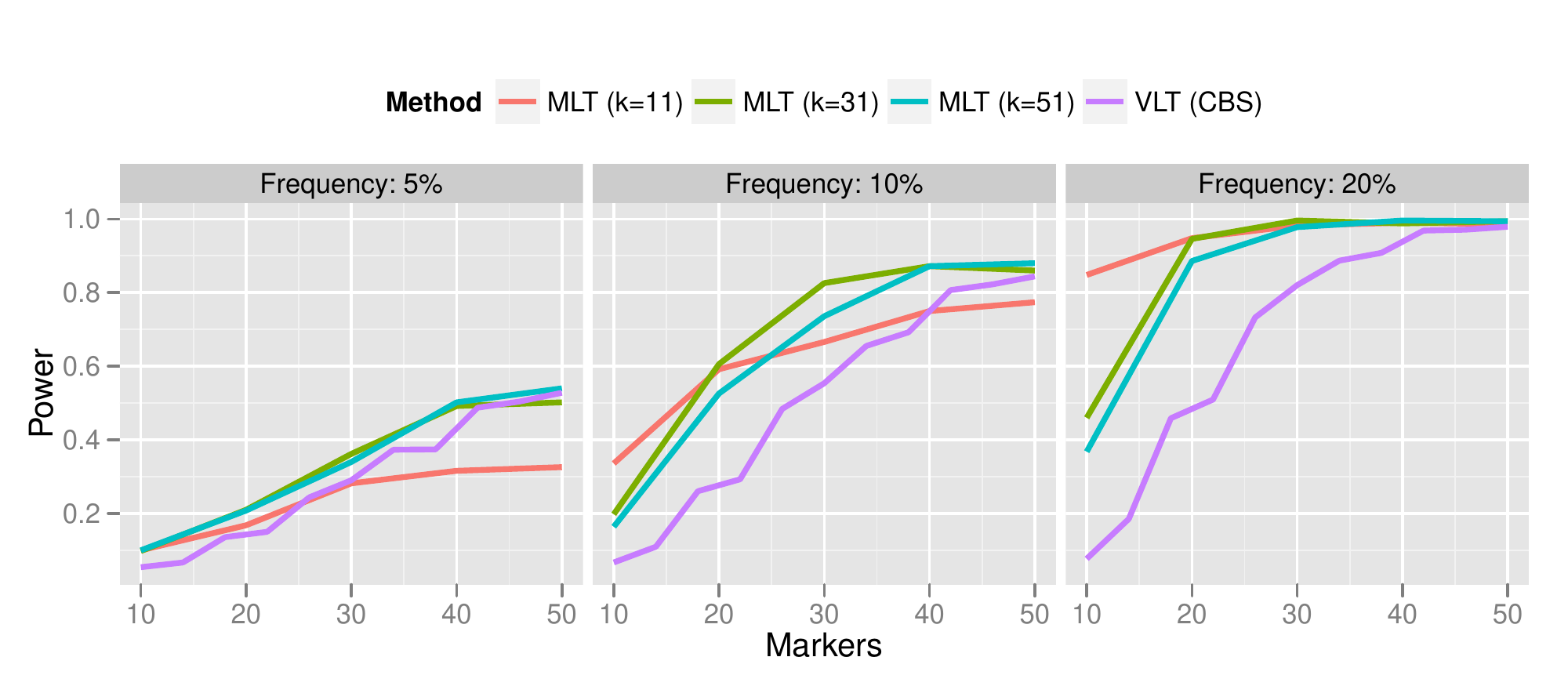}
 \caption{\label{Fig:vlt-mlt} Power comparison of variant-level testing (using CBS for CNV calling) with marker-level testing (using kernel-based aggregation).}
\end{figure}

For rare CNVs (5\% population frequency), the power of the variant-level approach and the aggregated marker-level approach are comparable.  However, for more common CNVs, the marker-level approach offers a substantial increase in power.  For the most part, this increase in power persists even when the bandwidth is misspecified.  Only when the bandwidth was much too small (selecting a 10-marker bandwidth for a 50-marker CNV) did the variant-level approach surpass marker-level aggregation.

Generally speaking, these results are consistent with the findings reported in \citet{Breheny2012}, who found that variant-level tests have optimal power relative to marker-level tests when CNVs are large and rare; conversely, marker-level tests have optimal power relative to variant-level tests when CNVs are small and common.  This is understandable given the limited accuracy of calling algorithms for small CNVs.

Comparing the results in Figure~\ref{Fig:vlt-mlt} with the results of \citet{Breheny2012}, who aggregated marker-level tests by applying CBS to the $p$-values as described in Section~\ref{Sec:exch}, we find that the kernel approach is a substantially more powerful method for aggregating marker-level tests than a change-point approach.  Specifically, Breheny et al. found that the change-point approach had very low power at 5\% frequency -- much lower than the variant-level approach.  On the other hand, in the same setting we find that the kernel approach is comparable to, and even slightly more powerful than, the variant-level approach.  Furthermore, as discussed in Section~\ref{Sec:exch}, a change-point analysis of marker-level tests also relies on exchangeability, which does not always hold.  Thus, the methods developed in this article are both more powerful and achieve better control over the FWER than the change-point analysis described in \citet{Breheny2012}.

A potential drawback of the kernel approach is the need to specify a bandwidth.  This makes the robustness of the method to bandwidth misspecification, as illustrated in Figure~\ref{Fig:vlt-mlt}, particularly important because in practice it is difficult to correctly specify the bandwidth {\em a priori}.  Indeed, it is possible that multiple CNVs associated with the outcome are present on the same chromosome and have different lengths.  A method that is not robust to bandwidth we be incapable of detecting both CNVs.  Generally speaking, a bandwidth of roughly 30 markers seems to provide good power over the range of CNV sizes that we investigate here.

\section{Discussion}

One drawback of the approach is the need to carry out permutation testing.  The kernel aggregation itself is very fast, but the need to carry out $\approx 1,000$ permutation tests for each marker may be highly computationally intensive (although easy to parallelize), depending the complexity of the marker-level test.  Ongoing research in our group is focusing on ways to speed up the approach described here with a model-based formulation that avoids the need for permutation testing.

The simulation studies of Section~\ref{Sec:sim} address a limited-scale version of a larger question: how do marker-level test aggregation and variant-level testing compare for chromosome-wide and genome-wide analysis?  This is an important question and deserves further study.  In general, multiplicity is a thorny issue for CNV analyses, as the true location of CNVs are unknown and can overlap in a number of complicated ways.  The issue of how many tests to carry out and adjust for is a challenging question for variant-level testing and a considerable practical difficulty in analysis.  In contrast, aggregation of marker-level results avoids this issue altogether.  We have shown that the proposed approach is both powerful at detecting CNV associations and rigorously controls the FWER at a genome-wide level --- two appealing properties.  However, future work analyzing additional studies using kernel aggregation and studying its properties in larger, more complex settings is necessary.

%An \texttt{R} package, \texttt{kbag}, containing an implementation of all the methods described in this article, will be made publicly available on CRAN (\url{cran.r-project.org}).

\subsection*{Acknowledgments}

We thank Brooke Fridley and Liang Li at the Mayo Clinic for contribution of the gemcitabine pharmacogenomic study data for this research.

\subsection*{Appendix}

\begin{proof}[Proof of Theorem~\ref{Thm:fwer}]
Let $\mathcal{P}$ denote the set of all possible permutations of $\{y_i\}$, $F_0$ the CDF of $T_{\max}$ over $\mathcal{P}$, and $F_0^{-1}$ its generalized inverse.  Also, let $\phi(\X,\y)=1$ if $T_{\max}(\X,\y) > F_0^{-1}(1-\alpha)$ and 0 otherwise.

Now, note that under the null hypothesis that $\x_i$ and $y_i$ are independent,
\as{P(\X,\y) &= \prod_i P(\x_i, y_i)\\
 &= \prod_i P(\x_i)P(y_i)\\
 &= P(\X,\y^*)}
for all $y^* \in \mathcal{P}$.  Thus, $\Ex_0 \phi(\X,\y^*)$ is a constant for all $\y^*$ and
\as{\Ex_0\left\{\phi(\X,\y)\right\} &= \frac{1}{n!} \sum_{\y^* \in \mathcal{P}} \Ex_0 \phi(\X,\y^*)\\
 &= \Ex_0 \frac{1}{n!} \sum_{\y^* \in \mathcal{P}} \phi(\X,\y^*)\\
 &\leq \alpha,}
where the term inside the expectation in the second line is less than or equal to $\alpha$ for all $\X$ and $\y$ by the construction of the test.
\end{proof}

\bibliographystyle{ims-nourl}

\end{document}